\RequirePackage[l2tabu,orthodox]{nag}
\documentclass[11pt,abstracton]{scrartcl}

\usepackage{authblk}

\usepackage[utf8]{inputenc}
\usepackage[T1]{fontenc}
\usepackage{lmodern}        
\usepackage{microtype}

\usepackage{amssymb}
\usepackage{amsmath}
\usepackage{amsthm}
\usepackage{bm}        
\usepackage{mathtools} 
\usepackage{xfrac} 
\usepackage[inline]{enumitem}
\setlist[itemize]{label=$\circ$}
\setlist[description]{labelindent=\parindent}

\usepackage{xifthen}
\usepackage{xspace}

\usepackage{tikz}
\usetikzlibrary{patterns}

\DeclarePairedDelimiter\paren{\lparen}{\rparen}
\DeclarePairedDelimiter\abs{\lvert}{\rvert}
\DeclarePairedDelimiter\set{\{}{\}}
\DeclarePairedDelimiterX\setc[2]{\{}{\}}{\,#1 \;\delimsize\vert\; #2\,}
\DeclarePairedDelimiterX\parenc[2]{\lparen}{\rparen}{\,#1 \;\delimsize\vert\; #2\,}
\newcommand{\cc}[1]{\ensuremath{\mathsf{#1}}}
\newcommand{\pp}[1]{\textup{#1}}
\newcommand{\op}[1]{\ensuremath{\operatorname{#1}}}

\newtheorem{theorem}{Theorem}[section]
\newtheorem{lemma}[theorem]{Lemma}
\newtheorem{corollary}[theorem]{Corollary}

\newcommand{\NP}{\cc{NP}}
\newcommand{\APX}{\cc{APX}}
\renewcommand{\P}{\cc{P}}

\newcommand{\FPTAS}{\cc{FPT\textrm{-}AS}}
\newcommand{\XP}{\cc{XP}}
\newcommand{\FPT}{\cc{FPT}}
\newcommand{\W}[1][]{%
 \cc{W}
  \ifthenelse{\isempty{#1}}%
    {}
    {[#1]}
}

\newcommand{\kmc}{$k$-\pp{Multicolored Clique}}

\newcommand{\CSP}{\pp{CSP}}

\newcommand{\MCSP}{\pp{MAX-CSP}}
\newcommand{\SAT}{\pp{CNF-SAT}}
\newcommand{\MSAT}{\pp{MAX-CNF-SAT}}
\newcommand{\MDNF}{\pp{MAX-DNF-SAT}}
\newcommand{\MPAR}{\pp{MAX-PARITY}}
\newcommand{\MAJ}{\pp{MAJORITY}}      \newcommand{\THR}{\pp{THRESHOLD}}
\newcommand{\MMAJ}{\pp{MAX-MAJORITY}} \newcommand{\MTHR}{\pp{MAX-THRESHOLD}}

\newcommand{\p}[2][]{
\ensuremath{
#2 
\ifthenelse{\isempty{#1}}
{^*} 
{(#1)} 
}}

\newcommand{\nd}{\op{nd}}
\newcommand{\cw}{\op{cw}}
\newcommand{\tw}{\op{tw}}
\newcommand{\mtw}{\op{mtw}}
\newcommand{\fvs}{\op{fvs}}
\newcommand{\vc}{\op{vc}}

\newcommand{\G}[1]{%
 \op{G}^*_{#1}
}

\newcommand{\occ}{\op{occ}}
\newcommand{\poly}{\op{poly}}
\newcommand{\andc}{\op{AND}}
\newcommand{\orc}{\op{OR}}
\newcommand{\parc}{\op{PARITY}}
\newcommand{\xorc}{\op{XOR}}
\newcommand{\majc}{\op{MAJORITY}}

\newcommand{\dotcup}{\mathbin{\dot\cup}}

\begin{document}

\title{Complexity and Approximability of Parameterized MAX-CSPs} 

\author[1]{Holger Dell}
\author[2]{Eun Jung Kim}
\author[3]{Michael Lampis}
\author[4]{Valia Mitsou\footnote{Supported by ERC Starting Grant PARAMTIGHT (No. 280152)}}
\author[5]{Tobias M\"omke\footnote{This research is supported by 
      Deutsche Forschungsgemeinschaft grant BL511/10-1}}

\affil[1]{
  Saarland University and Cluster of Excellence, Saarbr\"{u}cken%
  \\\texttt{hdell@mmci.uni-saarland.de}
}
\affil[2,3]{
  Université Paris Dauphine%
  \\\texttt{eunjungkim78@gmail.com, michail.lampis@dauphine.fr}
}
\affil[4]{
  SZTAKI, Hungarian Academy of Sciences, Budapest%
  \\\texttt{vmitsou@sztaki.hu}
}
\affil[5]{
    Saarland University, Saarbr\"{u}cken%
  \\\texttt{moemke@cs.uni-saarland.de}
}

%
%
%

\maketitle

\begin{abstract}

We study the optimization version of constraint satisfaction problems (Max-CSPs) 
in the framework of parameterized complexity;
the goal is to compute the maximum fraction of constraints that can be satisfied 
simultaneously.
In standard CSPs, we want to decide whether this fraction equals one.
The parameters we investigate are structural measures, such as the treewidth or 
the clique-width of the variable--constraint incidence graph of the CSP 
instance.

We consider Max-CSPs with the constraint types \andc, \orc, \parc, and \majc, 
and with various parameters~$k$, and we attempt to fully classify them into the 
following three cases:
\begin{enumerate}
  \item
    The exact optimum can be computed in $\FPT$ time.
  \item
    It is $\W[1]$-hard to compute the exact optimum, but there is a randomized 
    $\FPT$ approximation scheme (\FPTAS), which computes a 
    $(1-\epsilon)$-approximation in time $f(k,\epsilon) \cdot \poly(n)$.
  \item
    There is no \FPTAS\ unless $\FPT=\W[1]$.
\end{enumerate}
For the corresponding standard CSPs, we establish $\FPT$ vs.\ $\W[1]$-hardness 
results.

\end{abstract}

\section{Introduction}

Constraint Satisfaction Problems (\CSP s) play a central role in almost all
branches of theoretical computer science. Starting from \SAT, the prototypical
NP-complete problem, the computational complexity of \CSP s has been widely
studied from various points of view. In this paper we focus on two aspects of
\CSP\ complexity which, though extremely well-investigated, have mostly been
considered separately so far in the literature: parameterized complexity and
approximability. We study four standard predicates and contribute some of the
first results indicating that the point of view of approximability considerably
enriches the parameterized complexity landscape of \CSP s.

\paragraph{Parameterized Complexity.} For a \emph{parameterized} problem $P$, an instance of $P$ 
is a pair $(x,k) \in \Sigma^*\times \mathbb{N}$, where the second part $k$ of the instance is called the
\emph{parameter}. A parameterized problem $P$ is \emph{fixed-parameter tractable} ($\FPT$ in short) if there is an algorithm solving any input instance $(x,k)$ of $P$ in time $O(f(k)\cdot |x|^{O(1)})$ for some computable function $f$. Such an algorithm is called an \FPT-algorithm.

For two parameterized problems $P$ and $Q$, a \emph{parameterized reduction} from $P$ to $Q$ is an \FPT-algorithm which, given an instance $(x,k)$ of $P$, outputs an instance $(x',k')$ of $Q$ such that (i) $(x,k)$ is a yes-instance if and only if $(x',k')$ is a yes-instance, and (ii) $k' \le g(k)$ for some computable function $g$. The notion of paramterized reduction defines the hierarchy of parameterized complexity classes $$\FPT=W[0] \subseteq \W[1] \subseteq \W[2] \subseteq \cdots \subseteq \XP,$$ where each class is the family of problems admitting a parameterized-reduction to some basic problem. The central assumption in parameterized complexity is $\FPT \neq W[1]$. Further  parameterized complexity terminology used in this paper can be found in \cite{DF13}.

\paragraph{Parameterized \CSP s.} The vast majority of interesting \CSP s
are \NP-hard  \cite{S78,KSW97}. This has motivated the study of such problems from a
parameterized complexity point of view, and indeed this topic has attracted
considerable attention in the literature \cite{G06,S11,GS11,PRSW14,GS14,S11b}.
We refer the reader to \cite{samer2010constraint} where an extensive
classification of \CSP\ problems for a large range of parameters is given. In
this paper we focus on \emph{structurally} parameterized \CSP s, that is, we
consider \CSP s where the parameter is some measure of the structure of the
input instance.  The central idea behind this approach is to represent the
structure of the \CSP\ using a (hyper-)graph and leverage the powerful tools
commonly applied to parameterized graph problems (such as tree decompositions)
to solve the \CSP.

The typical goal of this line of research is to find the most general
parameterization of a \CSP\ that still remains fixed-parameter tractable ($\FPT$).
To give a concrete example for a very well-known \CSP, \SAT\ is $\FPT$ when
parameterized by the treewidth of its incidence graph\footnote{See the next
section for a definition of incidence graphs} \cite{SzeiderSAT2003} 
but it is $\W$-hard for more general
parameters such as clique-width \cite{OPS13}, or even the more restricted modular
treewidth \cite{PSS13}.  General (boolean) \CSP\ on the other hand, where the
description of each constraint is part of the input is known to be a harder
problem: it is already $\W[1]$-hard parameterized by the incidence treewidth, but
$\FPT$ parameterized by the treewidth of the primal graph \cite{S13a}. Thus,
parameterized investigations aim to locate the boundary where a \CSP\ jumps from
being $\FPT$ to being $\W$-hard. It is of course a natural question how we can deal
with the $\W$-hard cases of a \CSP\ once they are identified.

\paragraph{Approximation.} \CSP s also play a central role in the theory
of (polynomial-time) approximation algorithms \cite{T10,KS15,AK13}. In this context we
typically consider a \CSP\ as an optimization problem (\MCSP) where the goal is
to find an assignment to the variables that satisfies as many of the
constraints as possible. Unfortunately, essentially all non-trivial \CSP s are
hard to approximate ($\APX$-hard) from this point of view \cite{creignou1995dichotomy,KSW97}, even those
where deciding if an assignment can satisfy all constraints is in $\P$ (e.g.
\textsc{2\SAT} or \textsc{Horn SAT}). Thus, research in this area typically focuses on discovering
exactly the best approximation ratio that can be achieved in polynomial time.
Amazingly, for many natural \CSP s this happens to be exactly the ratio achieved
by a completely random assignment \cite{H01}. This motivates the question of
whether we can find natural cases where non-trivial efficient approximations
are possible.

\paragraph{Results.} In this paper we consider four different types of
\CSP s where the constraints are respectively \orc, \andc, \parc\ and \majc\
functions. Our approach follows, for the most part, the standard parameterized
complexity script: we consider the input instance's incidence graph and try to
determine the complexity of the \CSP\ when parameterized by various graph widths.
The new ingredient in our approach is that, in addition to trying to determine
which parameters make a \CSP\ $\FPT$ or $\W$-hard, we also ask if the optimization
versions of $\W$-hard cases can be well-approximated. We believe that this is a
question of special interest since, as it turns out, there are \CSP s for which
$\W$-hardness can be (almost) circumvented using approximation, and others which
are inapproximable.

More specifically, our results are as follows: for \orc\ constraints, which
corresponds to the standard \SAT\ (\MSAT) problem, we present a new hardness
proof establishing that deciding a formula's satisfiability is $\W$-hard even
if parameterized by the incidence graph's neighborhood diversity\footnote{Akin to neighborhood diversity is the \emph{twin-cover} number proposed in~\cite{Ganian11}. On bipartite graphs such as incidence graphs of CSPs, the twin-cover number is essentially the same as the vertex cover number: it differs only on a graph consisting of a single edge, in which the twin-cover number equals 0 while the vertex cover number is 1. Hence, we do not consider the twin-cover number separately as a structural parameter in this paper.}. Neighborhood
diversity is a parameter much more restricted than modular treewidth (already a restriction of clique-width) 
\cite{L12}, for which the strongest previously known $\W[1]$-hardness result
was known \cite{PSS13}. We complement this negative result with a strong
positive approximation result: there exists a randomized $\FPT$ Approximation 
Scheme
($\FPTAS$)\footnote{We follow here the standard definition of $\FPTAS$ given in
\cite{M08}.} for \MSAT\ parameterized by clique-width, that is, an algorithm
which for all $\epsilon>0$ runs in time $f(k,\epsilon)n^{O(1)}$ and returns an
assignment satisfying $(1-\epsilon)\mathrm{OPT}$ clauses. Thus, even though we
establish that solving \SAT\ exactly is $\W$-hard even for extremely restricted
dense graph parameters, \MSAT\ is well-approximable even in the quite
general case of clique-width. To the best of our knowledge, this is the first
approximation result of this type for a $\W$-hard \MCSP\ problem.

Recalling that \MSAT\ is $\FPT$ parameterized by the treewidth of the incidence
graph, we consider other problems for which the jump from treewidth to
clique-width could have interesting complexity consequences. We show that \MDNF\ 
and \MPAR, which are $\FPT$ parameterized by treewidth, exhibit two wildly 
different behaviors. On the one hand, the problem of maximizing the largest 
possible number of satisfied \parc\ constraints remains $\FPT$ even for dense 
parameters such as clique-width. On the other hand, by modifying our reduction for \SAT, we are
able to show not only that maximizing the number of satisfied \andc\
constraints is $\W[1]$-hard parameterized by neighborhood diversity, but
also that this problem cannot even admit an $\FPTAS$ (like \MSAT), unless
W[1]=$\FPT$. We recall that \parc\ and \andc\ constraints are similar in other
aspects: for example, for both we can decide in polynomial time if an
assignment satisfying all constraints exists.

Finally, we consider \CSP s with \majc\ constraints, that is, constraints which
are satisfied if at least half their literals are true. We give a reduction
establishing that this is an interesting case of a natural constraint type for
which deciding satisfiability is already $\W[1]$-hard parameterized by treewidth 
(we actually show $\W[1]$-hardness for the more restricted case of incidence 
feedback vertex set) and by neighborhood diversity.
We complement this negative result with two algorithmic results:
first, we show that the corresponding \MCSP\ is $\FPT$ parameterized by
incidence vertex cover. Then, we use this algorithm as a sub-routine to obtain
an $\FPTAS$ for incidence feedback vertex set. Both of these algorithmic
results also apply to the more general case of \THR\ constraints. We leave it
as an interesting open problem to examine if the
approximation algorithm for feedback vertex set can be extended to treewidth.

\tikzset{ref/.style = {
   text=gray
  }
}
\colorlet{hard}{red!60!white}
\colorlet{easy}{green}
\colorlet{medium}{cyan}

\begin{figure}[tp]
  \centering
  \begin{tikzpicture}[rounded corners]
    \node[fill=hard] (cw) at (0,0) {\p\cw};
    \node[fill=easy] (tw) at (1,1) {\p\tw};
    \node[fill=easy] (fvs) at (3,1) {\p\fvs};
    \node[fill=hard,label={[ref]above:L\ref{thm:CNFWhardness}}] (nd) at (2,0) 
    {\p\nd};
    \node[fill=easy] (vc) at (4,0) {\p\vc};
    \path[draw,latex-] (cw) -- (tw);
    \path[draw,latex-] (tw) -- (fvs);
    \path[draw,latex-] (fvs) -- (vc);
    \path[draw,latex-] (cw) -- (nd);
    \path[draw,latex-] (nd) -- (vc);

    \node at (2,-1) {CNF-SAT};
  \end{tikzpicture}
  \hspace{1cm}
  \begin{tikzpicture}[rounded corners]
    \node[fill=hard] (cw) at (0,0) {\p\cw};
    \node[fill=hard] (tw) at (1,1) {\p\tw};
    \node[fill=hard,label={[ref]above:T\ref{thm:majhard}}] (fvs) at (3,1) {\p\fvs};
    \node[fill=hard,label={[ref]above:T\ref{thm:majndhard}}] (nd) at (2,0) {\p\nd};
    \node[fill=easy] (vc) at (4,0) {\p\vc};
    \path[draw,latex-] (cw) -- (tw);
    \path[draw,latex-] (tw) -- (fvs);
    \path[draw,latex-] (fvs) -- (vc);
    \path[draw,latex-] (cw) -- (nd);
    \path[draw,latex-] (nd) -- (vc);

    \node at (2,-1) {MAJORITY-CSP};
  \end{tikzpicture}
  \caption{
    The parameterized complexity status of CNF-SAT and MAJORITY-CSP.  The boxes
    depict different parameterizations of each problem:
    \emph{red} means that the problem is \cc{W[1]}-hard and
    \emph{green} means that the problem is \cc{FPT}.
    Recall that DNF-SAT and PARITY-CSP are polynomial-time computable.
    An \emph{arrow} indicates the existence of an approximation-preserving reduction
    from the problem at the tail to the problem at the head, so for example, the
    arrow $\p{\fvs} \rightarrow \p{\tw}$ for CNF-SAT indicates that there is a
    reduction from CNF-SAT parameterized by $\p{\fvs}$ to CNF-SAT parameterized by
    $\p{\tw}$. In fact, the reductions we depict here are trivial since, for
    example, $\p{\fvs}$ is bounded by a function of $\p{\tw}$.
  }
\end{figure}
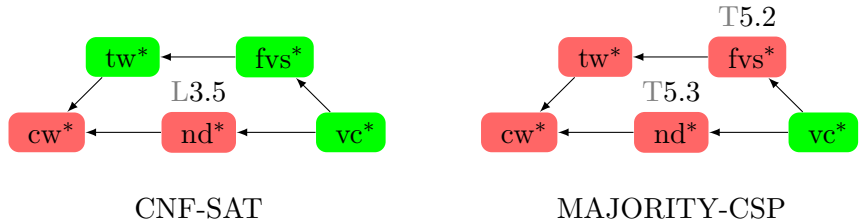

\begin{figure}[tp]
  \centering
  \begin{tikzpicture}[rounded corners]
    \node[fill=medium,label={[ref]above:T\ref{theorem: cw}}] (cw) at (0,0) 
    {\p\cw};
    \node[fill=easy] (tw) at (1,1) {\p\tw};
    \node[fill=easy] (fvs) at (3,1) {\p\fvs};
    \node[fill=medium] (nd) at (2,0) {\p\nd};
    \node[fill=easy] (vc) at (4,0) {\p\vc};
    \path[draw,latex-] (cw) -- (tw);
    \path[draw,latex-] (tw) -- (fvs);
    \path[draw,latex-] (fvs) -- (vc);
    \path[draw,latex-] (cw) -- (nd);
    \path[draw,latex-] (nd) -- (vc);

    \node at (2,-1) {MAX-CNF-SAT};
  \end{tikzpicture}
  \hspace{1cm}
  \begin{tikzpicture}[rounded corners]
    \node[pattern=north east lines,pattern color=medium] (cw) at (0,0) {\p\cw};
    \node[pattern=north east lines,pattern color=medium] (tw) at (1,1) {\p\tw};
    \node[fill=medium,label={[ref]above:T\ref{thm:thr_fvs}}] (fvs) at (3,1) {\p\fvs};
    \node[pattern=north east lines,pattern color=medium] (nd) at (2,0) {\p\nd};
    \node[fill=easy,label={[ref]above:T\ref{thm:thr_vc}}] (vc) at (4,0) {\p\vc};
    \path[draw,latex-] (cw) -- (tw);
    \path[draw,latex-] (tw) -- (fvs);
    \path[draw,latex-] (fvs) -- (vc);
    \path[draw,latex-] (cw) -- (nd);
    \path[draw,latex-] (nd) -- (vc);

    \node at (2,-1) {MAX-MAJORITY-CSP};
  \end{tikzpicture}
  \\[1cm]
  \begin{tikzpicture}[rounded corners]
    \node[fill=hard] (cw) at (0,0) {\p\cw};
    \node[fill=easy] (tw) at (1,1) {\p\tw};
    \node[fill=easy] (fvs) at (3,1) {\p\fvs};
    \node[fill=hard,label={[ref]above:T\ref{thm:dnf}}] (nd) at (2,0) {\p\nd};
    \node[fill=easy] (vc) at (4,0) {\p\vc};
    \path[draw,latex-] (cw) -- (tw);
    \path[draw,latex-] (tw) -- (fvs);
    \path[draw,latex-] (fvs) -- (vc);
    \path[draw,latex-] (cw) -- (nd);
    \path[draw,latex-] (nd) -- (vc);

    \node at (2,-1) {MAX-DNF-SAT};
  \end{tikzpicture}
  \hspace{1cm}
  \begin{tikzpicture}[rounded corners]
    \node[fill=easy,label={[ref]above:T\ref{thm:parity}}] (cw) at (0,0) {\p\cw};
    \node[fill=easy] (tw) at (1,1) {\p\tw};
    \node[fill=easy] (fvs) at (3,1) {\p\fvs};
    \node[fill=easy] (nd) at (2,0) {\p\nd};
    \node[fill=easy] (vc) at (4,0) {\p\vc};
    \path[draw,latex-] (cw) -- (tw);
    \path[draw,latex-] (tw) -- (fvs);
    \path[draw,latex-] (fvs) -- (vc);
    \path[draw,latex-] (cw) -- (nd);
    \path[draw,latex-] (nd) -- (vc);

    \node at (2,-1) {MAX-PARITY-CSP};
  \end{tikzpicture}
  \caption{
    The parameterized complexity status of MAX-CSP problems.
    The \emph{gray} labels above the boxes indicate the theorem in which we 
    establish the result; previously known results are displayed without 
    reference.
    \emph{Red} means that the problem is \cc{W[1]}-hard to compute exactly, and 
    there is no \FPTAS\ unless $\cc{FPT}=\cc{W[1]}$.
    \emph{Blue} means that the problem is \cc{W[1]}-hard to compute exactly, and 
    there is an \FPTAS.
    \emph{Green} means that the problem is \cc{FPT} to compute exactly.
    The \emph{blue/white} stripes mean that it's \cc{W[1]}-hard to compute 
    exactly, and it's open whether there is an \FPTAS.
  }
\end{figure}
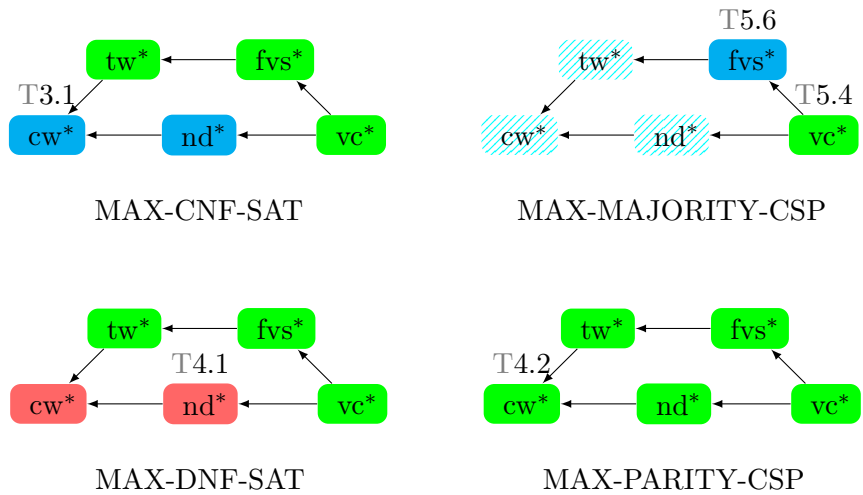

\section{Preliminaries}

\paragraph{Boolean \CSP.} A Boolean Constraint Satisfaction Problem, \CSP\ in short, $\psi$ is defined as a set $\set{ C_1, \ldots, C_m}$ of $m$
constraints over a set $X(\psi)=\set{x_1,\dotsc,x_n}$ of $n$ variables and
their negations.  Each constraint $C_i$ is regarded as a function of literals
(positive or negative appearances of variables) mapped to the set $\set{0,1}$,
where literals can take the values $0$ or $1$. Furthermore, we define $|C_j|$
to denote the arity of constraint $C_j$ (the number of literals that occur in
$C$) and $|\psi|=m$ the number of constraints in $\psi$. For simplicity, we
also write $l_i \in C_j$ for a literal $l_i$ and a constraint $C_j$ if $l_i$
appears in $C_j$.

We will be dealing with Boolean \CSP\ for four
well-studied Boolean functions: \orc\ constraints, \andc\ constraints, \parc\
(or \xorc) constraints and \majc\ constraints. We say that an assignment
$t:X\rightarrow\set{0,1}$ satisfies a constraint $C$ of type: 

\begin{itemize} 

\item \orc, if $\exists l_i \in C$, $t(l_i)=1$; 

\item \andc, if $\forall l_i \in C$, $t(l_i)=1$; 

\item \parc, if it satisfies some equation $\Sigma_{l_i\in C} t(l_i) = b$ (for
$b\in\{0,1\}$) modulo 2; 

\item \majc, if at least $\lceil\sfrac{|C|}{2}\rceil$ literals in $C$ are set
to $1$. More generally, we may consider \THR\ constraints, where a certain
threshold number of literals must be set to true to satisfy the constraint.

\end{itemize}

Let $\occ(\psi)=\sum_{C \in \psi} |C|$ be the total number of variable 
occurrences in $\psi$, that is, the total size of the formula.
For a variable $x$, we write $\psi_x$ for the set of all constraints $C\in\psi$ 
where $x$ occurs either positively or negatively; for the functions we consider
without loss of generality, no constraint contains both literals.  Thus, the total
number of occurrences of a variable $x$ is equal to $|\psi_x|$.

We are dealing also with \MCSP s, where given a set of constraints $\psi$, we
are interested in finding an assignment to the variables that maximizes the
number of satisfied constraints. The natural decision version of this problem
is, given a target $k$, decide whether there exists an assignment that
satisfies at least $k$ constraints. Clearly, the problem where we want to
decide whether we can satisfy all the constraints is a special case of the
above decision problem since we can set $k=m$, but in some cases we consider
this simpler decision version, particularly when we want to show hardness.

In the case of \orc\ constraints, the \CSP\ and \MCSP\ problems correspond to
the more widely known \SAT\ and \MSAT\ problems.  In this case we call the
constraints \emph{clauses}. When the constraint function is \andc, the \MCSP\
problem is called \MDNF.  In that case, the constraints are called
\emph{terms}. The problem \MPAR\ is also known as \pp{MAX-LIN-2} 
(satisfy a maximum number of given linear equations modulo 2).

\paragraph{Incidence graph and structural parameters.} For a \CSP\ $\psi$, the incidence graph $\G{\psi}$
is defined as the bipartite graph where we construct one vertex $v_i$ for each 
(unsigned) variable $x_i$ and one vertex $u_j$ for each constraint $C_j$ and 
connect $v_i$ with $u_j$ if $x_i \in C_j$. 

We are interested in parameterizations of the incidence graph $\p[\G{\psi}]{\op{p}}$ 
(or simply $\p{\op{p}}$ if $\G{\psi}$ is clear from the context), where
$\op{p}$ is a structural parameter of $\G{\psi}$.  We are mostly interested in
the two most widely studied graph parameters, treewidth $\p{\tw}$ and
clique-width $\p{\cw}$. The definitions of treewidth and clique-width are rather lengthy, and we refer the reader to standard parameterized
complexity textbooks for the definitions, for example~\cite{DF13}.

Another structural parameter we study is the \emph{incidence neighborhood diversity} denoted as $\p{\nd}$. 
A graph $G(V,E)$ has neighborhood diversity $\nd(G) = k$ if we can partition 
  $V$ into $k$ sets $V_1, \dots, V_k$ such that, for all $v\in V$ and all 
  $i\in\{1,\dots,k\}$, either $v$ is adjacent to every vertex in $V_i$ or it is 
  adjacent to none of them.
In other words, $\nd(G) = k$ if $V$ can be partitioned into $k$ modules that are 
either cliques or independent sets. We also investigate the complexity of CSPs parameterized by the vertex cover number $\p{\vc}$ and 
the feedback vertex set number $\p{\fvs}$ of the incidence graph, that is, the minimum number of vertices
that need to be deleted to make the graph edgeless and acyclic, respectively. 

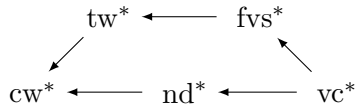
\begin{figure}[tp]
  \centering
  \begin{tikzpicture}
    \node (cw) at (0,0) {\p\cw};
    \node (tw) at (1,1) {\p\tw};
    \node (fvs) at (3,1) {\p\fvs};
    \node (nd) at (2,0) {\p\nd};
    \node (vc) at (4,0) {\p\vc};
    \path[draw,latex-] (cw) -- (tw);
    \path[draw,latex-] (tw) -- (fvs);
    \path[draw,latex-] (fvs) -- (vc);
    \path[draw,latex-] (cw) -- (nd);
    \path[draw,latex-] (nd) -- (vc);
  \end{tikzpicture}
  \caption{\label{fig:parameters}%
    The structural parameters we study and their relationships.
    For example, the arrow between $\p\cw$ and $\p\tw$ means that if the 
    treewidth is bounded, then the clique-width is bounded as well --- more 
    precisely, there is a monotone computable function $f:\mathbf N\to\mathbf N$ 
    so that $\p\cw\le f(\p\tw)$.
    On the other hand, \p\tw\ and \p\nd\ as well as \p\fvs\ and \p\nd\ cannot be 
    bounded by each other in general.
  }
\end{figure}

\section{\SAT\ and \MSAT}

In this section, we consider one of the most fundamental problems in computer 
science: the satisfiability problem for CNF formulas, which can be viewed as a 
constraint satisfaction problem where the only allowed constraints are clauses, 
that is, \orc s of literals.
An optimal solution for \MSAT\ can be computed in $\FPT$ when parameterized by 
the treewidth~$\p{\tw}$ of the incidence graph~\cite{AR11}, and hence \SAT\ can 
be solved in the same time.
When parameterized by the clique-width $\p{\cw}$ of the incidence graph, all 
known exact algorithms for \SAT\ and \MSAT\ run in $\XP$ time~\cite{SS13,STV14}.
Moreover, we do not expect these problems to be in \FPT\ since they are both 
$\W[1]$-hard parameterized by~$\p{\cw}$~\cite{PSS13}.%

In Section~\ref{sec:fptas}, we construct an approximation scheme for \MSAT\ that 
runs in $\FPT$ time.
Intuitively, our algorithm works as follows: given a formula $\phi$ with
`small' incidence clique-width, we first examine the formula to see if it
contains many or few `large' clauses. If the formula contains relatively few
large clauses, then we simply disregard them. We then know that the incidence
graph does not contain `large' bi-cliques, so by a theorem of Gurski and Wanke
\cite{GW00} the remaining formula has small treewidth and we can solve the problem.
If on the other hand the original formula contains almost exclusively large
clauses, then we observe that we can rely on a random assignment to satisfy
almost everything.

The hard part of our algorithm is then how to deal with the general case of
formulas that may contain clauses of varying arities, for which we use a
combination of the ideas for the two basic cases. In particular, after locating
and deleting a negligibly small set of `medium' clauses, we use a counting
argument to find a set of variables that appear almost exclusively in large
clauses. By setting these variables randomly we satisfy almost all large
clauses, and we can then use treewidth to handle the remaining instance.

In Section \ref{sec:nd_whard}, we explore a class of CSP instances that is 
smaller than the class of bounded incidence clique-width instances;
our goal is to understand which incidence graph parameter is responsible for the 
transition from \FPT\ to $\W[1]$.
To this end, we have to look for a graph parameter that is bounded by a function 
of $\p\cw$ (where the problem is hard) but can leave the $\p\tw$ unbounded (where it's 
\FPT).
In fact, \cite{PSS13} shows that the problem is $\W[1]$-hard parameterized by 
the modular incidence treewidth $\p\mtw$, which lies between 
$\p\cw$ and $\p\tw$.%
\footnote{
  A CNF formula has \emph{bounded modular incidence treewidth} if its incidence graph has bounded treewidth 
  after merging   variable/clause modules into a single vertex. Here, a variable/clause module is a set of 
  vertices, corresponding to variables/clauses respectively, with same neighborhood outside of the set.
  In fact, the reduction in \cite{PSS13} constructs a formula whose incidence 
  graph has small feedback vertex set after contracting modules.
}
We study the incidence neighborhood diversity $\p{\nd}$, which is 
incomparable to $\p{\tw}$; however, $\p{\mtw}$ is bounded when $\p{\nd}$ is.
We prove that \SAT\ remains $\W[1]$-hard parameterized by $\p{\nd}$.

Formulas whose incidence graph has neighborhood diversity at most~$k$ seem very 
restrictive: there are only at most~$k$ variable- and clause-types, where all 
variables of the same type belong to the same clauses and all clauses of the 
same type involve the same variables.
This class of formulas is a subset of formulas with $\p{\mtw}\le k$ because 
contracting all modules leaves a graph of \emph{order} at most $k$, which 
trivially has treewidth at most $k$.

\subsection{Approximation Algorithm parameterized by 
  clique-width}\label{sec:fptas}

\begin{theorem}\label{theorem: cw}
  There is a randomized algorithm such that, for every $\epsilon >0$ and given a CNF formula $\psi$ with $n$ 
  variables, $m$ clauses, and incidence clique-width $\cw$, runs in time $f(\epsilon,\cw) 
  \cdot \poly(n+m)$, and outputs a truth assignment that satisfies 
  at least $(1-\epsilon)\cdot \mathrm{OPT}$ clauses in expectation.
\end{theorem}

We formulate the following basic lemma about probability distributions.
\begin{lemma}\label{lemma:partition}
For all $\epsilon,L > 0$ there is a $c=c(\epsilon,L)>0$ such that all $c' \ge c$ and all 
sequences $p_1,\dots,p_{c'} \geq 0$ with $\sum_{i=1}^{c'} p_i \leq 1$ have an index $d\leq 
c/L$ with the property
\[
  p_{[d,L\cdot d]} \doteq \sum_{j=d}^{L\cdot d} p_j < \epsilon\,.
\]
\end{lemma}
\begin{proof}
  Let $\epsilon,L>0$.
  We set $c=c(\epsilon,L)$ below.
  Assume for contradiction that $p_{[d,L\cdot d]}\geq \epsilon$ holds for all 
  $d\in [1,c/L]$.
  If there are $1/\epsilon + 1$ disjoint intervals $[a_1,L\cdot 
  a_1],\dots,[a_{1/\epsilon+1},L\cdot a_{1/\epsilon+1}]\subseteq [1,c]$,
  then we arrive at a contradiction with the fact that the $p_i$'s are 
  non-negative and sum to at most one.
  Clearly there exists a constant $c=c(\epsilon,L)$ such that $1/\epsilon + 1$ 
  disjoint intervals of the form $[a,La]$ fit into $[1,c]$.
  This proves the claim.
\end{proof}

For an arbitrary given $\epsilon>0$, we fix $L=\epsilon^{-4}$.
We use Lemma~\ref{lemma:partition} as follows:
For a CNF formula $\psi$, we define $p_i$ as the fraction of clauses of 
size~$i$, that is,
\begin{align*}
  p_i \doteq \frac{\Bigl|\setc[\Big]{ C\in\psi}{|C|=i}\Bigr|}{\abs{\psi}}\,.
\end{align*}
Then Lemma~\ref{lemma:partition} gives us a number $d\leq c(\epsilon)$ such that 
the total fraction of clauses whose size is between $d$ and $\epsilon^{-4}d$ is 
bounded by $\epsilon$.
It is now natural to partition all clauses into short, medium, and long clauses.
More precisely, we define $\psi=\psi^{<d}\dotcup\psi^{[d,D]}\dotcup\psi^{>D}$ 
for $D=\epsilon^{-4}d$ as follows:
\begin{align*}
  \psi^{<d}&\doteq\setc[\Big]{C\in\psi}{ \abs{C} < d }\,,\\
  \psi^{[d,D]}&\doteq\setc[\Big]{C\in\psi}{ d \leq \abs{C} \leq D}\,, 
  \text{and}\\
  \psi^{>D}&\doteq\setc[\Big]{C\in\psi}{ \abs{C} > D}\,.
\end{align*}
An immediate corollary to Lemma~\ref{lemma:partition} is thus that we can choose 
$d\leq c(\epsilon)$ in such a way that $|\psi^{[d,D]}|\leq \epsilon |\psi|$.
\begin{corollary}\label{corollary: middle}
For all $\epsilon > 0$ there is some $c=c(\epsilon)>0$ such that all CNF 
formulas~$\psi$ have some $d=d(\epsilon)\in[1,c]$ with
$
  \abs{\psi^{[d,\epsilon^{-4}d]}} \leq \epsilon \cdot |\psi|\,.
$
\end{corollary}

If $\psi^{[d,D]}=\emptyset$ holds for $D=\epsilon^{-4}d$ and 
$d\in[1,c(\epsilon)]$, we say that $\psi$ is \emph{$\epsilon$-well separated}.
We call $\psi$ \emph{$\epsilon'$-balanced} if, in addition, we have 
$|\psi^{<d}|\ge\epsilon'm$ and $|\psi^{>D}| \ge \epsilon' m$.
\begin{lemma}
\label{lem:sparse-left}
Let $\psi$ be an $\epsilon$-well separated formula (and thus $V=V(\psi^{<d}) \cup 
V(\psi^{>D})$).

Then, for each subset $\hat{\psi} \subseteq \psi^{>D}$ with $|\hat{\psi}| > 
\epsilon^2 m$, there is a variable $y$ such that
$|\psi^{<d}_y| \le \epsilon^2 |\hat{\psi}_y|$.
\end{lemma}
That is, for every set $\hat\psi$ that contains a significant fraction of long 
clauses, there is a variable that occurs $\abs{\hat\psi_y}$ times in $\hat\psi$, 
but only at most an $\epsilon^2$-fraction of that in the short clauses.
\begin{proof}
  Let $\hat\psi\subseteq\psi^{>D}$ with $\abs{\hat\psi} > \epsilon^2m$.
  Note that the total number of literal occurrences in $\hat\psi$ is 
  $\occ(\hat{\psi}) > D \cdot \epsilon^2 \cdot m = \epsilon^{-2} d m$.
  In contrast, $\occ(\psi^{<d}) < d m$.
  Now suppose that there was no variable $y$ with the claimed properties, that 
  is, suppose that every variable $y$ satisfies $\abs{\psi^{<d}_y} > \epsilon^2 
  \abs{\hat\psi_y}$.
  Then the total number of variable occurrences in $\psi^{<d}$ can be bounded from below as 
  follows:
  \[
    \occ(\psi^{<d})=\sum_{y} \abs{\psi^{<d}_y} > \sum_{y} \epsilon^2 
    |\hat{\psi}_{y}| = \epsilon^2 \occ(\hat{\psi}) > d\cdot m\,.
  \]
  This yields a contradiction and thus proves the claim.
\end{proof}

\begin{proof}[Proof of Theorem~\ref{theorem: cw}]
The algorithm $A$ works as follows.
Let $\epsilon' = \epsilon^2$, and we assume w.l.o.g.\ that $\epsilon < 1/8$.
Given a CNF formula~$\phi$, we compute an $\epsilon'$-well separated formula 
$\psi$ by dropping all clauses in $\phi^{[d,D]}$.
Corollary \ref{corollary: middle} guarantees that the fraction of deleted clauses 
is bounded by $\epsilon'$.
If $\psi$ is not $\epsilon/2$-balanced, we discard the smaller side (with fewer clauses) 
and only handle the larger one:
If $\psi^{<d}$ is the larger side, we compute an optimal assignment for 
$\psi^{<d}$ in FPT time, by using the result of Gurski and Wanke~\cite{GW00}.
This way the total number of unsatisfied clauses is at most $\epsilon m/2$, and together with
the unsatisfied clauses due to applying Corollary~\ref{corollary: middle}, the total number of unsatisfied clauses is smaller than  $\epsilon m$. Since $\mathrm{OPT} > m/2$, we get the approximation guarantee.

If $\psi^{>D}$ is the larger side, we use a random assignment.
This way, at most $\epsilon m/2$ clauses from $\psi^{<d}$ are violated, and in expectation
at most a $2^{-D}$ fraction of clauses from $\psi^{>D}$ are violated. Since $2^{-D}$ is smaller than $\epsilon/4$,
we conclude that -- together with unsatisfied clauses due to applying Corollary~\ref{corollary: middle} -- at least 
$(1-\epsilon)m$ clauses are satisfied in expectation.

This finishes the analysis of unbalanced formulas, and in the remaining proof we may assume that 
$\psi$ is $\epsilon/2$-balanced.
To handle this case, we determine a set of variables $Y$ such that
\begin{itemize}
\item there are at most $\epsilon m/4$ short clauses with variables from $Y$ and
\item there are at most $\epsilon^2 m$ long clauses that contain $\le 1/\epsilon$ variables from $Y$.
\end{itemize}

Before we construct $Y$, let us verify that the properties of $Y$ imply the correctness of the theorem.
Our algorithm computes a satisfying assignment of the short clauses without variables from $Y$, again using the result
of Gurski and Wanke~\cite{GW00}. Afterwards it assigns values uniformly at random to the variables in $Y$.

There are at most $\epsilon' m = \epsilon^2 m$ unsatisfied clauses due to applying Corollary~\ref{corollary: middle},
$\epsilon m/4$ short clauses clauses that we did not consider when satisfying clauses from $\psi^{<d}$,
and $\epsilon^2 m$ clauses from $\psi^{>D}$ that we did not consider in the random assignment. Additionally,
in expectation there are less than $2^{-1/\epsilon}m$ clauses left unsatisfied from the 
remaining $|\psi^{>D}| - \epsilon^2 m$ clauses from $\psi^{>D}$. Since we assumed that $\epsilon < 1/8$, the theorem follows.

To construct the set $Y$, we iteratively apply Lemma~\ref{lem:sparse-left} with the parameter $\epsilon/4$. 
Initially, we set $\hat\psi = \psi^{>D}$. In each iteration, we identify a variable $y$ according to the lemma and add the variable to $Y$.
In the subsequent iterations, we mark $y$ to be inactive and handle it as if it was not contained in any clause.
Whenever we identify a clause $C$ that has at least $1/\epsilon$ inactive variables (i.\,e., variables from $Y$), 
we remove $C$ from $\hat\psi$.
We continue this process until $|\hat\psi| \le \epsilon^2$.
Note that applying Lemma~\ref{lem:sparse-left} for $\epsilon/4$ but having an $\epsilon'$-well separated formula ensures that at all times, 
all clauses in $\hat\psi$ have sufficiently many literals to apply Lemma~\ref{lem:sparse-left}. 
Therefore the process terminates and the generated set $Y$ has the aimed-for properties since $|Y| \le m/\epsilon$.
\end{proof}

\subsection{Hardness parameterized by neighborhood 
  diversity}\label{sec:nd_whard}

A \emph{constraint} on $r$ variables is a relation $R\subseteq\{0,1\}^r$.
We define the unary constraints $U_0=\{0\}$ and $U_1=\{1\}$, which corresponds 
to clauses $(\neg x)$ and $(x)$, respectively.
We define the equality $=$ and disequality $\neq$ constraints on two groups of 
Boolean variables $x = x_1 \dots x_n$ and $y= y_1\dots y_n$ in infix notation in 
the usual way:
For an assignment $a$ to the $x$- and $y$-variables, we say that $a \models x=y$ 
if and only if, for all $i\in[n]$, we have $a(x_i)=a(y_i)$, that is, the 
assignment sets $x_i$ to the same value as $y_i$; as usual, $x\neq y$ is 
interpreted as the negation of $x=y$.

\begin{lemma}\label{thm:CNFWhardness}
  CNF-SAT parameterized by $\p{\nd}$ is $\W[1]$-hard, where $\p{\nd}$ is the 
  neighborhood diversity of the input's incidence graph.
\end{lemma}
\begin{proof}
  We devise an FPT-reduction from \kmc\ to CNF-SAT.
  Given a $k$-partite graph $G$, whose parts $V_1,\dots,V_k$ all have the same 
  size~$n$, we construct $k$ groups of variables $x_1,\dots,x_k$, which together 
  are supposed to represent a $k$-clique in $G$, should one exist.
  Each group $x_i$ consists of $\log n$ Boolean variables and represents the 
  supposed clique's vertex in the part $V_i$.
  Without loss of generality, we assume that $\log n$ is an integer.
  
  Starting from the empty CNF formula, we construct a formula $\phi$ on the 
  $x$-variables as follows.
  First choose, for each $i\in[k]$, an arbitrary bijection $\mathrm{bin}_i : V_i 
  \to \{0,1\}^{\log n}$ that maps any vertex $u\in V_i$ to its binary 
  representation $\mathrm{bin}(u)$; for convenience, we drop the index~$i$.
  For each $i,j\in[k]$ with $i < j$, and for each non-edge $(u,v)\not\in 
  E(V_i,V_j)$ between $V_i$ and $V_j$, we add the following constraint 
  $C_{i,j,u,v}$ to $\phi$:
  \[
    x_ix_j \neq \mathrm{bin}(u)\mathrm{bin}(v)\,.
  \]
  Clearly, this constraint excludes exactly one of the $2^{2\log n}$ possible 
  assignments to $x_ix_j$, and so it can be written as an OR of literals of the 
  $x$-variables.
  In the end, $\phi$ is a CNF formula with $|E(G)|$ clauses.

  For the completeness of the reduction, let $v_i\in V_i$ for all $i\in[k]$ be 
  such that $v_1,\dots,v_k$ is a clique in~$G$.
  For all $i\in[k]$, set $x_i = \mathrm{bin}(v_i)$.
  This assignment satisfies all constraints: for all $(u,v)\not\in E(V_i,V_j)$,
  we have that $\mathrm{bin}(v_i)\mathrm{bin}(v_j)\neq 
  \mathrm{bin}(u)\mathrm{bin}(v)$ because $(v_i,v_j)$ is an edge and $(u,v)$ is 
  not, and $\mathrm{bin}$ is a bijection.

  For the soundness of the reduction, let $a_1,\dots,a_k\in\{0,1\}^{k\log n}$ be 
  a satisfying assignment of $\phi$.
  For each $i\in[k]$, let $v_i$ be the unique vertex in $V_i$ for which 
  $\mathrm{bin}(v_i) = a_i$.
  Let $i,j\in[k]$ with $i < j$.
  Since the assignment satisfies all constraints of $\phi$, it must be the case 
  that $(v_i,v_j)$ is an edge in $G$.
  For if it was a non-edge, its corresponding constraint in $\phi$ would have 
  excluded the assignment $a_ia_j$ for $x_ix_j$.
  Hence $v_1,\dots,v_k$ is a clique in~$G$.

  It remains to argue that the neighborhood diversity of the incidence graph 
  of~$\phi$ is at most $k + \binom{k}{2}$.
  The modules of the incidence graph are the variable group $x_h$ for each 
  $h\in[k]$ and the clause group $\{C_{i,j,u,v}\}$ for each $i,j\in [k]$ with $i 
  < j$.
  Indeed, the incidence graph between $x_{h}$ and $C_{i,j,*,*}$ is a bipartite 
  clique if $h\in\{i,j\}$, and otherwise it is an independent set.

  We constructed an FPT-reduction from the W[1]-complete problem Multicolored 
  Clique to CNF-SAT parameterized by $\p{\nd}$, which finishes the proof of the 
  theorem.
\end{proof}

\section{From Treewidth to Clique-width}

In the previous section, we have seen that \MSAT\ is fixed-parameter tractable 
when parameterized by $\p{\tw}$, which is a sparse graph parameter, and it is 
W[1]-hard to compute exactly and has an $\FPTAS$ when parameterized by $\p{\cw}$, 
which is a dense graph parameter.
In this section we observe that the transition from sparse to dense parameters
has different effects on the complexity of \MCSP, depending on which types of 
constraints are allowed.

By modifying our reduction for \SAT\ we show that \MDNF, the problem of 
maximizing the number of satisfied \andc\ constraints is W[1]-hard parameterized by 
$\p{\nd}$.
Furthermore, because the maximum number of constraints that could be
satisfied in our reduction is also bounded by some function of the parameter,
we show that the problem does not have an $\FPTAS$ unless FPT=W[1].
Thus, while \MDNF\ is FPT parameterized by $\p{\tw}$, it does not even appear to 
have an FPT approximation scheme when parameterized by $\p{\nd}$.

\begin{theorem} \label{thm:dnf}
  Suppose that there exists an $\FPTAS$ which, given $\epsilon >0$ and an 
  instance $\phi$ of \MDNF, computes a $(1-\epsilon)$-approximate solution and 
  runs in time $f(\p{\nd},\epsilon)\cdot \poly(n)$, where $\p{\nd}$ is the 
  neighborhood diversity of the incidence graph of $\phi$.
  Then $\FPT = \W[1]$.
\end{theorem}
\begin{proof}
  We devise an \FPT-reduction from \kmc\ that is similar to the one in our proof 
  of Theorem~\ref{thm:CNFWhardness}.

  Given a $k$-partite graph $G$ whose parts $V_1,\dots,V_k$ have size $n$ each, 
  we construct $k$ groups $x_1,\dots,x_k$ of $\log n$ variables each, and 
  $\binom{k}{2}$ groups of \andc\ constraints $C_{i,j,u,v}$ for each integers 
  $i,j\in[k]$ with $i < j$ and edge $(u,v) \in E(V_i,V_j)$ between $V_i$ and 
  $V_j$:
  \[
    x_ix_j = \mathrm{bin}(u)\mathrm{bin}(v)
  \]
  Here, $\mathrm{bin}(u) \in \{0,1\}^{\log n}$ is some binary representation of 
  $u\in V_i$.
  Note that the constraint $C_{i,j,u,v}$ is satisfied by exactly one of the  
  $2^{2\log n}$ assignments to the variables $x_ix_j$, and so this constraint 
  can be written as an AND of literals of these variables.
  Apart from producing $\phi$, the reduction also sets the approximation 
  parameter $\epsilon = k^{-2}$ so that $(1-\epsilon) \binom{k}{2} > 
  \binom{k}{2} - 1$ holds.
  The neighborhood diversity of the incidence graph of $\phi$ is at most 
  $k+\binom{k}{2}$ by a completely analogous argument as in the proof of 
  Theorem~\ref{thm:CNFWhardness}.

  We now prove that an $\FPTAS$ for $\MDNF$ would allow us to distinguish 
  whether $G$ has a $k$-clique or not.
  For the completeness of the reduction, let $G$ have a $k$-clique 
  $v_1,\dots,v_k$.
  Then the assignment $x_1 \dots x_k = 
  \mathrm{bin}(v_1) \dots \mathrm{bin}(v_k)$ satisfies the $\binom{k}{2}$ 
  constraints $C_{i,j,v_i,v_j}$ for each $i,j\in[k]$ with $i<j$.
  Thus the assumed $\FPTAS$ will return a solution that satisfies at least 
  $(1-\epsilon)\binom{k}{2} > \binom{k}{2}-1$ constraints.
  Since the number of satisfied constraints is an integer, it must thus be at 
  least $\binom{k}{2}$ (and in fact is equal to $\binom{k}{2}$ in this case).
  For the soundness, if $G$ has no $k$-clique, then at most $\binom{k}{2} - 1$ 
  constraints can be simultaneously satisfied, and so the assumed $\FPTAS$ can 
  not return a solution whose value is larger than that.

  Finally, note that the overall algorithm above solves the Multicolor Clique 
  problem in time $f(\p{\nd},\epsilon) \poly(n) \le g(k)\poly(n)$, which is 
  \FPT; thus an $\FPTAS$ for $\MDNF$ would imply that $\FPT=\W[1]$.
\end{proof}

When parameterized by \p{\tw}, \MSAT\ and \MDNF\ are both \FPT, and when 
parameterized by a dense graph parameter, such as $\p{\nd}$, the former problem 
is hard but approximable while the latter problem is hard even to approximate.
We next consider natural constraint types where the corresponding \CSP s stay 
\FPT\ both for sparse as well as dense incidence graph parameters.
\MPAR\ wants to find an assignment that satisfies the maximum number of linear 
equations modulo two.
While deciding whether there is an assignment that satisfies all equations is in 
P (by Gauss elimination), the maximization version is a typical APX-hard 
problem~\cite{H01}.
Here we show that computing the optimal solution of \MPAR\ is FPT, regardless of 
whether the parameter is the treewidth or the clique-width of the incidence 
graph.
Our intuition for why \MPAR\ appears to be so much easier than \SAT\ is that 
negations are (almost) irrelevant, and so the incidence graph seems to capture 
most of the structure relevant to the complexity of the \CSP\ instance.

\begin{theorem} \label{thm:parity}
  Given an instance $\phi$ for \MPAR, we can find an optimal solution in time 
  $f(\p{\cw})|\phi|^{O(1)}$, where $\p{\cw}$ is the clique-width of the 
  incidence graph of $\phi$.
\end{theorem}
\begin{proof}
  We rely on the meta-theorem of \cite{CMR00} that all problem which can be 
  expressed in the optimization version of CMSO$_1$ can be solved exactly in 
  linear time.
  In the logic CMSO$_1$, we are allowed to express problems via first-order 
  formulas with additional second-order quantifiers that are only allowed to 
  quantify over subsets of the universe of the input structure, and with 
  counting constraints that stipulate the cardinalities of these sets modulo a 
  constant number.

  To express \MPAR\ in CMSO$_1$, observe that when given a linear equation 
  $\sum_i l_i = b$ over $\operatorname{GF}(2)$, where $b\in\{0,1\}$ and the 
  $l_i$ are literals (either $x_i$ or $1-x_i = 1+x_i$), we may view it 
  equivalently as a constraint of the form $\sum_i x_i = b'$, where all the 
  $x_i$ are variables and $b'=b$ if and only if the original constraint contains 
  an even number of negated literals on the left hand side.

  Having performed the above pre-processing we can now express our problem in
  CMSO$_1$.
  The structure we construct is a bipartite directed graph $G$ with the 
  bipartition $L\dotcup R$, which is represented by an edge relation $E$.
  For each linear equation $\sum_i x_i = b$, we introduce a vertex $\ell$ 
  in~$L$, for each variable we introduce a vertex~$x_i$ in~$R$, and we set 
  $E(\ell,x_i)$ if and only if $x_i$ appears in equation~$\ell$.
  Moreover, we have the unit constraints $U_0$ and $U_1$, and we set $U_b(\ell)$ 
  if and only if $b$ is the right-hand side of $\ell$.
  
  The CMSO$_1$ formula that we construct is looking for the largest set 
  $S\subseteq L$ of constraint-vertices such that there exists a set of 
  variables $X\subseteq R$ which satisfies the following: every vertex $\ell\in 
  S$ where $U_b(\ell)$ holds has a number of neighbors in $X$ that is equal 
  to~$b$ modulo two, for every $b\in\{0,1\}$.
  By construction, the maximum such set~$S$ corresponds to the maximum set of 
  linear equations that can be satisfied simultaneously by an assignment 
  represented by~$X$.
\end{proof}

\section{Majority and Threshold CSPs}

In this section we deal with CSPs where each constraint is a \MAJ\ or a \THR\
constraint.
In this problem, each constraint is supplied with an integer value
$t$ (the threshold) and it is satisfied if and only if at least $t$ of its
literals are set to true. \MAJ\ is the special case of this predicate where $t$
is always equal to half the arity of each constraint.
We denote the resulting satisfiability problem with \MAJ\ and \THR, and the 
resulting \MCSP s with \MMAJ\ and \MTHR.

\MAJ\ and \THR\ constraints are of course some of the most natural and
well-studied predicates in many contexts: for example, \MCSP\ for such
constraints contains the complexity of finding an assignment that satisfies as
many inequalities as possible in a 0-1 Integer Linear Program whose
coefficients are in $\{-1,0,1\}$.  This problem, sometimes called
\textsc{Maximum Feasible Subset} has been well-studied in the literature
\cite{ERRS09,AK98,AK95}. \MAJ\ constraints also play a central role in learning theory \cite{FGRW12,GR09}
and in hardness of approximation \cite{DMN13}.

\subsection{Hardness of exact algorithms}

We consider the problem whether a CSP with \THR\ constraints is satisfiable.
This problem is \NP-complete.
We parameterize the problem by the size $\p{\fvs}$ of the smallest feedback 
vertex set, or by the neighborhood diversity $\p{\nd}$ of the instance's 
incidence graph.
As we will see, these parameterized problems turn out to be hard, even for the 
special case of \MAJ\ constraints.
Thus, neither dense nor sparse incidence graph parameters appear to put the 
problem in \FPT.

In order to ease notation in the upcoming proofs, we note that \THR-constraints 
are quite expressive.
For example, they can express clauses
$\ell_1 \vee\dots\vee \ell_d$
since this is the same as requiring that at least one of the $d$ literals is 
true.
Similarly, stipulating that at least $d$ literals be true is the same as a term 
$\ell_1\wedge\dots\wedge\ell_d$.
Finally, we can stipulate AT-MOST-ONE$(\ell_1,\dots,\ell_d)$, that is, that at 
most one of the literals is set to true, by using the \THR-constraint that at 
least $d-1$ of the literals $\neg\ell_1,\dots,\neg\ell_d$ are true.
\begin{theorem}
  \label{thm:thrhard}
  \THR\ parameterized by $\p{\fvs}$ is $\W[1]$-hard.
\end{theorem}
\begin{proof}
We reduce from the \kmc\ problem.
Let $G(V_1,\ldots , V_k,E)$ be a $k$-partite graph with $\abs{V_i}=n$ for each 
$i\in\{1,\dots,k\}$, and we regard the vertices of $V_i$ to be identified with 
integers from 1 to $n$.
Let $E_{ij}$ be the edge set between $V_i$ and $V_j$ for every $1\leq i<j\leq 
k$.
We construct the output formula $\phi$ of $\THR$ using the following gadgets for 
every $i$ and every pair $1\leq i<j\leq k$:

\begin{itemize}
\item {\em Vertex selection gadget for $V_i$:} For each vertex $\ell \in V_i$, 
  we create a sequence of $\ell$ variables $P_{\ell}$ for $\phi$ and name the 
  first and the last variable in the sequence $p_{\ell}$ and $p'_{\ell}$ 
  respectively (in particular, if $\ell=1$, we have $p_1=p'_1$).
  For every two consecutive variables $y,z \in P_{\ell}$, we add an 
  OR-constraint $C=(y\vee \neg{z})$.
  These constraints guarantee that, in any satisfying assignment and for every 
  $\ell$, if $p'_\ell$ is set to true, then all variables in $P_\ell$, including 
  $p_\ell$, are set to true as well.

  Given the $n$ variable sets $P_1,\ldots ,P_n$, we add two constraints $X_i$ 
  and $Y_i$ as follows:
  \begin{align*}
    X_i=\text{AT-MOST-ONE}(p_\ell : \ell \in V_i)
    \quad \text{and} \quad
    Y_i=\text{AT-LEAST-ONE}(p'_\ell :\ell \in V_i)
    \,.
  \end{align*}
  The constraint $Y_i$ enforces at least one $p'_{\ell}$ is set to true, which 
  propagates through all variables in $P_\ell$ to $p_{\ell}$.
  Conversely, the constraint $X_i$ requires that at most one $p_{\ell}$ be set 
  to true.
  Hence, any satisfying assignment will have exactly one $\ell\in V_i$ for which 
  $p_{\ell}$ is set to true; moreover, all variables in the set $P_{\ell}$ are 
  set to true, and all the variables in $P_{\ell'}$ for $\ell'\in 
  V_i\setminus\{\ell\}$ are set to false.

\item {\em Edge selection gadget for $E_{ij}$:}
  The gadget for $E_{ij}$ is similar to that for $V_i$.
  For each edge $e=(k,\ell) \in E_{ij}$ with $k\in V_i$ and $\ell \in V_j$, we 
  create a sequence of $n+1-\ell$ variables $Q_{e}$ and name the first and the 
  last variable in the sequence $q_{e}$ and $q'_e$ respectively.
  For every two consecutive variables $y,z \in Q_{e}$, we add an OR-constraint 
  $C=(y\vee \neg{z})$.
  Given the variable sets $Q_e$ for each $e\in E_{ij}$, we add two constraints 
  $X_{ij}$ and $Y_{ij}$ as follows:
  \begin{align*}
    X_{ij}=\text{AT-MOST-ONE}(q_e : e \in E_{ij})
    \quad \text{and} \quad
    Y_{ij}=\text{AT-LEAST-ONE}(q'_e :e \in E_{ij})
    \,.
  \end{align*}
  By the same argument as for the vertex selection gadget, we have exactly one 
  variable $q_e$ set to true, in which case all variables of $Q_e$ are set to 
  true and all variables in $Q_{e'}$ are set to false for all $e'\neq e$.

\item {\em Incidence verification gadget between $V_i$ and $E_{ij}$:}
  For every edge $e=k\ell \in E_{ij}$ with $k\in V_i$,  we create an 
  OR-constraint $C_{ke}=(p_k\vee \neg{q_{e}})$.
  This guarantees that any satisfying assignment that sets $q_e$ to true also 
  sets $p_k$ to true.

\item {\em Incidence verification gadget between $V_j$ and $E_{ij}$:}
  We add two constraints $C_{ij}$, $C'_{ij}$ and their $t$-values as
  \begin{align*}
    C_{ij}
    &=
    \text{AT-LEAST}_{n+1}
    \paren*{
      \bigcup_{e\in E_{ij}} Q_e \cup \bigcup_{\ell\in V_j} P_\ell
    }\,, \text{and}\\
    C'_{ij}
    &=
    \text{AT-MOST}_{n+1}
    \paren*{
      \bigcup_{e\in E_{ij}} Q_e \cup \bigcup_{\ell\in V_j} P_\ell
    }\,.
  \end{align*}
  Any satisfying assignment will thus set exactly $n+1$ variables to true among 
  all the variables in the $Q_e$ sets for $e\in E_{ij}$ and the $P_\ell$ sets 
  for $\ell\in V_j$.
  In particular, there must be a natural number $\ell$ such that $\ell$ of the 
  true variables are in the $P$-sets, and $n+1-\ell$ of the true variables are 
  in the $Q$-sets.
  By the constraints on the $P$-variables, the set of true variables is equal to 
  exactly one of the sets $P_{\ell'}$.
  Since $\abs{P_{\ell'}} = \ell'$ for all $\ell'\in\{1,\dots,n\}$, we have that 
  $\ell>0$ and that the selected set is exactly $P_\ell$.
  Moreover, exactly one of the sets $Q_e$ is selected due to the constraints on 
  the $Q_e$-variables.
  Since $n+1 = |Q_e| + |P_\ell| = |Q_e| + \ell$ holds if and only if $e=k\ell$ 
  for some for some $k\in V_i$, the selected set $Q_e$ must correspond to an 
  edge~$e$ incident to $\ell\in V_j$.
\end{itemize}

For the completeness of the reduction, let $S$ be a $k$-clique of $G$. Then we 
set all $\ell$ variables of $P_{\ell}$ from the vertex selection gadget for 
$V_i$ to true if $\ell\in V_i$ is chosen for the $k$-clique $S$. Also we set all 
$n+1-\ell$ variables of $Q_e$, where $e=(k,\ell)$,  from the edge selection 
gadget for $E_{ij}$ to true if $k\in V_i$ and $\ell \in V_j$ are chosen for $S$.  
It is not difficult to check that this assignment satisfies $\phi$.

Conversely, suppose $\phi$ has a satisfying assignment. From the above argument, 
for each $V_i$ there is exactly one variable set to true among all $p_{\ell}$'s 
for  $1\leq \ell\leq n$. Likewise, for each $E_{ij}$, exactly one variable $q_e$ 
is {\em selected}. From the property of the incidence verification gadget 
between $V_i$ and $E_{ij}$, whenever $q_e$ from the edge selection gadget for 
$E_{ij}$ is {\em selected}, $p_k$ such that $k$ is the vertex in $V_i$ incident 
with $e$ in $G$ must be set to true. Lastly, by the property of incidence 
verification gadget between $V_j$ and $E_{ij}$, that $q_e$ with $e=(k,\ell)$ is 
selected implies that exactly $\ell$ variables are set to true in the vertex 
selection gadget for $V_j$. From the property of the vertex selection gadget, 
this means that $P_{\ell}$ is selected and thus $p_{\ell}$ is set to true.  
Hence, for every $e=(k,\ell)\in E_{ij}$, whenever $q_e$ is set to true, both 
$p_k$ and $p_{\ell}$ are set to true. By selecting the vertices of $G$ 
corresponding to $k\in V_i$ such that $p_k$ is set to true, we find a $k$-clique 
of $G$. 

Consider the set of constraints \[F=\{X_i,Y_i:1\leq i\leq k\}\cup 
  \{X_{ij},Y_{ij},C_{ij},C'_{ij}:1\leq i<j\leq k\}.\] Consider the graph 
obtained by deleting $F$ from the incidence graph of $\phi$. Each component 
involves a variable set $P_k$,  for some vertex $k \in V_i$, and $Q_e$ for all 
edges $e=(k,\ell)$ such that $\ell \in V_j$ for some $j>i$. It is easy to see 
that each component is in fact a graph obtained by subdividing edges of a star. 
Hence the incidence graph of $\phi$ has a feedback vertex set of size $O(k^2)$. 
This completes the proof.
\end{proof}

\begin{theorem}\label{thm:majhard}
  $\MAJ$ parameterized by $\p{\fvs}$ is $\W[1]$-hard, where $\p{\fvs}$ is the 
  minimal size of a feedback vertex set of the instance's incidence graph.
\end{theorem}
\begin{proof}
  We devise an \FPT-reduction from \THR\ to \MAJ\ that keeps $\p{\fvs}$ the 
  same.
  The result then follows from Theorem~\ref{thm:thrhard}.

  The reduction's input is a CSP $\phi$ with threshold constraints, and the 
  output is a CSP~$\phi'$ with majority constraints.
  We transform threshold into majority constraints by adding fresh variables, 
  some of which are forced to be either true or false.

  Without loss of generality, we may assume that the arity of each constraint $C$ in $\phi$ is even; if $C$ has odd number of literals, we add a dummy variable $y$ to $C$ as a positive literal and add a new constraint $(\neg y)$ with threshold equal to 1. Notice that this does not change $\p(\fvs)$. Clearly, the output instance is satisfiable if and only if the original instance is satisfiable.

  For each constraint $C$ that has the threshold
  $t(C) =\sfrac{|C|}{2} + d$ for some integer~$d$, we 
  add $2|d|$ fresh dummy variables $y_1,\dots,y_{2|d|}$ as positive literals to obtain a constraint $C'$ 
  whose threshold we set to $t(C')=\sfrac{|C'|}{2}=\sfrac{|C|}{2}+|d|$; 
  thus, $C'$ is a majority constraint, and we include it into $\phi'$.
  If $d>0$, we additionally add the constraint $(\neg y_i)$ for each $i\in[d]$ 
  to force the dummy variables to be false. If $d<0$, we add the constraints $(y_i)$ for each  $i\in[d]$ 
  to force the dummy variables to be true.

  For the correctness of the reduction, note that any satisfying assignment 
  of~$\phi'$ sets any dummy variable $y$ to false if $\phi'$ contains the 
  constraint $(\neg y)$ since $y$ occurs exactly once as a positive literal. Likewise, any satisfying assignment of~$\phi'$ sets any dummy variable $y$ to true if $\phi'$ contains the 
  constraint $(y)$.
  Thus, any assignment to the old variables satisfies a constraint $C$ if and 
  only if setting the $y$-variables as just specified sets at least half of the 
  literals in $C'$ to true.
  Thus $\phi$ is satisfiable if and only if $\phi'$ is.

  Adding new variables to a constraint corresponds to adding leaves to the 
  corresponding constraint vertex in the incidence graph, which does not create 
  any new cycles.
  Adding a single unit clause for a $y$-variable also does not create new cycles 
  as this just adds a leaf to some leaf.
  Thus $\p{\fvs}(\phi)=\p{\fvs}(\phi')$, and our reduction has the claimed 
  properties.
\end{proof}

\begin{theorem}\label{thm:majndhard}
  \THR\ and \MAJ\ parameterized by the incidence neighborhood diversity 
  $\p{\nd}$ are $\W[1]$-hard.
\end{theorem}
\begin{proof}
  First note that OR-constraints are special cases of \THR-constraints, and so 
  the hardness of \THR\ follows from Lemma~\ref{thm:CNFWhardness}, where we establish 
  the hardness of \SAT.
  Next we reduce from \SAT\ to \MAJ\ so that the the neighborhood 
  diversity of the new instance is linearly bounded by that of the original instance.
  For this, let $\varphi$ be a CNF-formula and consider its incidence graph with 
  neighborhood diversity~$k$.
  Let $G_1,\dots,G_t$ be the $t\le k$ modules corresponding to constraints of 
  $\varphi$.
  That is, every constraint in $G_i$ depends on the same set of vertices, and in 
  particular, every constraint in $G_i$ has the same arity~$a_i$.
  Now for each $i$, we add a set~$Z_i$ of $a_i-1$ fresh dummy variables as well 
  as the constraints $(z)$ forcing all $z\in Z_i$ to true.
  We convert each constraint of $G_i$ to a \MAJ-constraint by adding the 
  variable set $Z_i$ to every constraint in $G_i$.
  Since the $z\in Z_i$ are forced to true, each constraint always contains 
  $a_i-1$ variables set to true among its $2 a_i - 1$ variables.
  Thus the new \MAJ-constraint is satisfied if and only if the old OR-constraint 
  was.
  The reductions adds at most~$k$ groups~$Z_i$ of variables, which form modules 
  in the incidence graph, and so the incidence neighborhood diversity of the 
  output instance is at most~$2k$.
\end{proof}

\subsection{Exact Algorithm parameterized by vertex cover}

Motivated by the negative result of Theorem \ref{thm:majhard} we now
investigate the complexity of \MAJ\ for more restricted parameters. The first
parameter we consider is the vertex cover of the incidence graph. This is a
natural, though quite restrictive, parameter which is often considered for
problems which are W-hard for treewidth.

\begin{theorem}\label{thm:thr_vc}
 \MTHR\ parameterized by the incidence vertex cover $\p{\vc}$ is~$\FPT$.
\end{theorem}

\begin{proof}
  Given a \CSP\ $\phi$ with \THR-constraints over a variable set $X$, and a 
  size-$k$ vertex cover~$S$ of the incidence graph, we define
  $S_X, S_\phi \subseteq S$ to be the vertices of $S$ corresponding to variables 
  and constraints, respectively.
  The algorithm starts by branching into $2^{\abs{S_X}}\le 2^k$ cases, 
  corresponding to truth assignments~$\sigma$ on the variables in~$S_X$.
  Since~$S$ is a vertex cover, all variables of constraints $C\in \phi\setminus 
  S_\phi$ are contained in~$S_X$, and so we can compute how many of these 
  constraints are satisfied by $\sigma$; let's call this number $N_\sigma$.
  After fixing the variables in~$S_X$ and removing the constraints in 
  $\phi\setminus S_\phi$, what remains is a \CSP\ $\phi'$ with at most $k$ 
  constraints.
  Let $N'_{\sigma}$ be the maximum number of constraints that can be 
  simultaneously satisfied in $\phi'$.
  Then $\max_\sigma N_\sigma + N'_\sigma$ is the maximum number of constraints 
  of $\phi$ that can be simultaneously satisfied.
  Thus it remains to compute $N'_{\sigma}$ in \FPT-time.

  Let $\phi'$ be a \MCSP\ with $k$ constraints.
  We reduce it to $2^k$ instances~$\phi''$ of the standard \CSP-version of \THR, 
  by guessing for each constraint whether or not it is satisfied by an optimal 
  solution.
  The size of the largest subset of constraints that can be simultaneously 
  satisfied is then precisely the optimal value of the \MCSP-instance.

  Finally, to solve \THR, we further reduce $\phi''$ to an integer linear 
  program (ILP) with $3^k$ variables.
  This suffices since solving ILPs is \FPT\ when parameterized by the number of 
  variables~\cite{lenstra1983integer}.
  Let $k$ be the number of constraints of $\phi''$.
  We associate with each variable $x$ of $\phi''$ the vector 
  $v(x)\in\set{-1,0,1}^k$ with
  \begin{align*}
    v(x)_i
    &=
    \begin{cases}
      1 & \text{if $x$ occurs positively in the $i$-th constraint,}
      \\
      -1 & \text{if $x$ occurs negatively in the $i$-th constraint, and}
      \\
      0 & \text{if $x$ does not occur in the $i$-th constraint.}
    \end{cases}
  \end{align*}
  For each vector $v\in\set{-1,0,+1}^k$, we add a variable $\ell_v$ to the ILP and 
  the constraints $0\le \ell_v \le B_v$, where $B_v\in\mathbf N$ is the number 
  of variables~$x$ with $v(x)=v$.
  The variable $\ell_v$ is supposed to indicate how many of the variables of 
  type~$v$ are set to true.
  Finally, we translate each constraint $C_i\in\phi''$ to a linear equation.
  Let $T_i$ be the threshold of the constraint, and let $n_i$ be the number of 
  variables that occur as negative literals in~$C_i$.
  Then we add the following inequality to the ILP:
  \begin{align}\label{eq:ILP}
    \paren[\Big]{
      \sum_{\substack{v\in\set{-1,0,1}^k\\v_i=1}}\ell_v
    }
    -
    \paren[\Big]{
      n_i
      -
      \sum_{\substack{v\in\set{-1,0,1}^k\\v_i=-1}}\ell_v
    }
    &
    \ge
    T_i
    \,.
  \end{align}

  To prove the completeness of this reduction, let $\phi''$ have a satisfying 
  assignment~$\sigma$.
  Then we set $\ell_v$ to be the number of variables~$x$ with $v(x)=v$ such that 
  $\sigma(x)=1$.
  This satisfies $0\le\ell_v\le B_v$.
  All other linear constraints are generated from some constraint~$C_i$ of 
  $\phi''$, which is satisfied by~$\sigma$.
  The first term in the difference of~\eqref{eq:ILP} is exactly the number of 
  variables $x$ that are set to true under $\sigma$ and that occur positively in 
  $C_i$, and the second term is the number of variables~$x$ that are set to 
  false and that occur negatively in~$C_i$; thus, the left-hand side is the 
  number of literals set to true and the right-hand side is the threshold 
  of~$C_i$, and so $\eqref{eq:ILP}$ holds.

  For the soundness, assume there is a solution $\paren{\ell_v}_v$ for the ILP.
  Then we construct an assignment~$\sigma$ as follows: For each~$v$, arbitrarily 
  select $\ell_v$ of the $B_v$ variables with $v(x)=v$ to true, and set the 
  others to false.
  Then for each constraint~$C_i$ of $\phi''$, the linear constraint 
  \eqref{eq:ILP} guarantees that $C_i$ is satisfied.
  This finishes the correctness proof of the final reduction; overall, we solve 
  \MTHR\ in \FPT-time when parameterized by \p\vc.
\end{proof}

\subsection{Approximation Algorithm parameterized by feedback vertex set}

The results of Theorem \ref{thm:majhard} naturally pose the following question:
can we evade the W-hardness of \MAJ\ by designing an $\FPTAS$ for the problem?
In this section, though we do not resolve this question, we give some first
positive indication that this may be possible. We consider \MMAJ\ parameterized
by the incidence graph's feedback vertex set. This is a natural,
well-studied parameter that generalizes vertex cover but is a restriction of
treewidth. It is also connected to the concept of back-door sets to acyclicity,
which is well-studied in the parameterized CSP literature \cite{OPS13,GS12}.

Observe that approximating this CSP is non-trivial, since \MMAJ\ with
constraints of arity two already generalizes \textsc{Max-2SAT}, and is hence
APX-hard.
On the other hand, \MMAJ\ can easily be 2-approximated by considering
any assignment and its negation. Hence, the natural goal here is an
approximation ratio of $(1-\epsilon)$.
Using Corollary~\ref{thm:thr_vc} as a sub-routine we achieve this with an 
\FPTAS.

\begin{lemma}\label{lem:thr_acyclic}
There exists a polynomial-time algorithm that, given an instance $\phi$ of \MTHR\ whose incidence graph is acyclic, 
finds an optimal assignment to $\phi$. Furthermore, there is a truth assignment satisfying at least half of the constraints simultaneously.
\end{lemma}
\begin{proof}
      We assume that the incidence graph is connected since the application of an algorithm for each connected components
      will lead to an optimal solution to the original instance. We may assume that there is no  isolated vertex or a constraint vertex whose threshold is equal to zero: 
      if one exists, we can remove the corresponding variable or constraint without changing the set of optimal solutions.
      
      Consider the incidence graph as a tree 
      each rooted at a variable vertex.        
      Pick a variable vertex $v$ which
	is farthest from the root. Let $C_1,\ldots , C_p$ (possibly $p=0$) be the children of $v$. 
	All of $C_1,\ldots , C_p$ are exactly the constraints of the form $(v)$ or $(\neg v)$ since 
	they are leaves of the tree and thus contain no variable other than $v$. Furthermore, there is no other constraint having $v$ as
	the sole literal: the only remaining constraint, if one exists, is the parent $C$ of $v$ and $C$ must be incident with another variable vertex 
	since $C$ cannot be the root. We set the variable $v$ either to true or false so as to 
	maximize the number of satisfied constraints among $C_1,\ldots , C_p$. If there is a tie, that is, if $p=0$ or there are equal number of constraints
	of the form $(v)$ and $(\neg v)$, then (i) if $v$ is not the root, we set $v$ so that $v$ appears in its parent $C$ as a true literal, (ii) otherwise, we set $v$ arbitrarily.  
	After setting the assignment to $v$, we remove $v$ and decrease by one the threshold of all constraints in which $v$ appears as a true literal 
	(and remove all isolated vertices and  constraint vertices whose threshold becomes zero). We repeat the procedure until the incidence
	graph becomes empty.
	
	Clearly, the above procedures finds an optimal assignment when $v$ is the root. Suppose that $v$ is not the root and we set $v$ to true.
	Let $\phi'$ be
	the resulting instance. We claim that for any assignment to $\phi'$, additionally setting $v$ to true satisfies 
	as many constraints of $\phi$ as setting $v$ to false. Indeed, the claim holds if $C$, the parent of $v$, is satisfied by the extended assignment. 
	If $C$ is not satisfied by the extended assignment, notice that this is because there are strictly more constraints of the form $(v)$ than those of 
	the form $(\neg v)$. Hence, the claim holds in this case as well. A symmetric argument holds when we set $v$ to false. It follows that 
	the above algorithm finds an optimal assignment to $\phi$ in polynomial time. 
	
	Notice that at every step we choose a variable vertex $v$ 
	and set the assignment at $v$, at least half of the constraints which are removed at the end of the step are satisfied.
	The second part of the statement follows.
	\end{proof}

\begin{theorem}\label{thm:thr_fvs}
  There exists an $\FPTAS$ which, given $\epsilon >0$ and an instance $\phi$ of 
  \MMAJ, computes a $(1-\epsilon)$-approximate solution and runs in time 
  $f(\p{\fvs},\epsilon)\cdot \poly(n)$, where $\p{\fvs}$ is the size of the 
  smallest feedback vertex set of the incidence graph of $\phi$.
\end{theorem}

\begin{proof}
If $|\phi|\le (1+\sfrac{2}{\epsilon})k$, the the number of constraints is 
bounded by a function of $k$ and $\epsilon$, and we can use the \FPT-algorithm from Theorem
~\ref{thm:thr_vc}.

Therefore, we assume that $|\phi| > (1+\sfrac{2}{\epsilon})k$. With a similar argument as in the 
      proof of Theorem~\ref{thm:thr_vc} we can consider that the 
      $\fvs(\G{\phi})$ contains only constraint vertices (that is, we guess the 
      assignment of variables in the feedback vertex set). We now proceed by 
      simply deleting these constraints from the instance, and let $\phi'$ be the resulting instance. 
      Note that the incidence graph of $\phi'$ is acyclic. We invoke the polynomial-time algorithm of 
      Lemma~\ref{lem:thr_acyclic} to find an optimal assignment to $\phi'$. 
      
      Call the produced solution 
      $\textrm{SOL}(\phi)$ and the optimal solution $\textrm{OPT}(\phi)$.  From the optimality of the solution on $\phi'$, we have: 
      $\textrm{OPT}(\phi)\ge \textrm{OPT}(\phi')\ge \textrm{OPT}(\phi) - k$. Now, observe that $\textrm{OPT}(\phi')\ge 
      |\phi'|/2$ by Lemma~\ref{lem:thr_acyclic}. Therefore $\textrm{OPT}(\phi) > 
      k/{\epsilon}$ which gives $\frac{\textrm{OPT}-k}{\textrm{OPT}} > 
      1-\epsilon$.

\end{proof}

\bibliographystyle{abbrv}
\bibliography{param-sat}

\end{document}